\documentclass[final,3p,twocolumn]{elsarticle} 	
\journal{Systems \& Control Letters}
\usepackage{amsmath,amssymb,bm,bbm,mathrsfs,amscd,color,hyperref}
\usepackage{amsthm,subfigure}
\usepackage{dsfont, color}
\usepackage{graphicx}
\usepackage{epsfig}
\def\qed{\hfill \vrule height 7pt width 7pt depth 0pt}
\def\beq{\begin{equation}}
\def\eeq{\end{equation}}

\newtheorem{theorem}{Theorem}

\newtheorem{proposition}[theorem]{Proposition}

\theoremstyle{remark}

\newcommand{\ds}{\displaystyle}

\newcommand{\ba}{\begin{array}}
\newcommand{\ea}{\end{array}}
\newcommand{\tcr}{}
\newcommand{\tcb}{}
\renewcommand{\l}{\left}\renewcommand{\r}{\right}

\newcommand{\be}{\begin{equation}}
\newcommand{\ee}{\end{equation}}
\newcommand{\eps}{\varepsilon}

\newcommand{\mc}{\mathcal}

\newcommand{\ov}{\overline}

\newcommand{\1}{\mathbbm{1}}

\newcommand{\R}{\mathbb{R}}

\newcommand{\summ}{\sum\limits}
\renewcommand{\P}{\mathbb{P}}

\def\1{\mathds{1}}

\def\R{\mathbb{R}}

\def\P{\mathbb{P}}

\def\diag{{\rm diag}\,}

\begin{document}
\begin{frontmatter}
\title{From local averaging to emergent global behaviors:\\the fundamental role of network interconnections
\tnotetext[titlenote]{
This work is dedicated to the memory of Jan C.~Willems.}
}



\author{Giacomo Como\tnotetext[authornote]{The first author is a member of the excellence centers LCCC and ELLIIT. His research has been supported by the Swedish Research Council through  a junior research grant. 
} and Fabio Fagnani}
\address{G.~Como is with the Department of Automatic Control, Lund University, BOX118, SE-22100 Lund, Sweden {\tt\small giacomo.como@control.lth.se}.\\ F.~Fagnani is with the Lagrange Department of Mathematical Sciences, Politecnico di Torino, Corso Stati Uniti 24,  10129, Torino, Italy, {\tt\small fabio.fagnani@polito.it}.} %









\begin{abstract}
Distributed averaging is one of the simplest and most widely studied network dynamics. Its applications range from cooperative inference in sensor networks, to robot formation, to opinion dynamics. 
A number of fundamental results and examples scattered through the literature are gathered here and some original approaches and generalizations are presented, emphasizing the deep interplay between the network interconnection structure and the emergent global behavior.

\end{abstract}
\begin{keyword}  Network dynamics\sep interconnections \sep emergent behaviors \sep distributed averaging \sep consensus \sep polarization \sep electrical networks



\end{keyword}

\end{frontmatter}
\section{Introduction}

One of the core concepts in the behavioral approach to systems and control developed by Jan Willems in the '80s is that of interconnection \cite{Willems1}. Encompassing the traditional notion of feedback interconnection on which classical input/output control theory is based, the behavioral approach allows for defining interconnections of systems at a more primitive level, as intersections of solution sets of the evolution equations, without the need for specific flow diagrams. As Jan used to repeat, what is an input and what is an output  is a matter of the application. This idea of going beyond the classical input/output formalism proved fruitful in applications, e.g., in coding theory, where Willems' study of minimal state space realizations \cite{Willems2} laid the foundations of trellis representations which are the basic tool for the design of efficient decoding algorithms.

More recently, the study of network dynamics is showing deep cultural analogies with the {\it ansatz} of the behavioral approach. Network dynamics entail a large number of (relatively) simple systems coupled together along the architecture of a graph. The overall dynamical system can thus be seen as the interconnection of these atomic devices. It does not make much sense to classify {\it a priori} interconnection signals as input or outputs, rather they are variables coupling the systems, possibly sensor measurements, state positions, epidemic states, and it is often impossible to say who is influenced by whom. The emergence of global behaviors such as synchronization, information fusion, polarization, and diffusion is one of the distinctive features of these complex interconnected systems. Such global behaviors can in fact be seen as the result of the local interactions and of the interconnection graph structure. 

This paper focuses on a particularly simple and well studied class of network dynamics: distributed averaging systems.  \tcr{\cite{johnthes,OlshevskyTsitsiklis:2011}} These are linear network dynamics exhibiting many interesting collective behaviors, such as synchronization and transition phenomena. Their applications range from inferential sensor network algorithms \tcr{\cite{xiao}}, to network vehicle formation \tcr{\cite{ali}}, to models for opinion dynamics \tcr{\cite{degroot}}. Most of the behavioral approach developed by Jan was in fact focused on linear systems: he used to say that linear systems are sufficiently rich from a theoretical viewpoint and yet containing a huge variety of applications. Keeping models as simple as possible was central in Jan's approach to science.

Using classical results from the Perron-Frobenius theory of non-negative matrices, we first present an asymptotic analysis of the linear averaging dynamics on arbitrary interconnection graphs. As expected, the graph topology plays a crucial role in shaping the emergent global behavior. While it is well known that all states reach an asymptotic consensus on connected graphs, Theorem \ref{theo convergence} of Section \ref{sec:asymp} analyzes the case of a general graph and shows that the asymptotic state of every agent in the network turns out to be a convex combination of the consensus reached by the sink connected components (i.e., components  with no outgoing links). The weights of such convex combination have several useful interpretations. They can be seen as hitting probabilities of the dual Markov chain generated by the same averaging matrix or, when the graph is undirected, as voltages of an electrical circuit with suitable boundary conditions on the nodes belonging to the sink components, as explained in Section \ref{sec:electrical}. {While analogous electrical interpretations are well known in Markov chain theory \cite{DoyleSnell:84,LevinPeresWilmer}, they have received relatively minor attention in the distributed averaging literature, with a few exceptions. In particular, to our knowledge, Theorem \ref{theorem-resistances} has not appeared elsewhere in this generality. } 

A relevant case in the applications is when the sink components all consist of single nodes ---called stubborn nodes--- that never change their state, e.g., playing the role of opinion leaders in social networks, or anchor nodes in robot formation control. The final part of the paper is dedicated to a deeper understanding of how the asymptotic state is distributed within the network in the presence of such stubborn nodes. It turns out that ---depending on the stubborn nodes' centrality and the graph connectivity--- quite different phenomena can emerge ranging from polarized to homogeneous equilibrium configurations. \cite{Acemoglu.etal:2013} In the polarized case, nodes tend to cluster in subfamilies and converge to values very close to that of a particular stubborn agent, whereas in the homogeneous regime most of the nodes tend to get close to a consensus on a value which is a convex combination of the stubborn node values. In Section \ref{sec:polarizedvshomo}, we present these phenomena through an example where the transition between the two regimes can be analyzed in detail. We then recall more general results appeared in the literature.


We gather here some notational conventions. The transpose of  $A$ is denoted by $A'$; $\1$ is the all-$1$ vector; $\1_{\mc A}$ is the vector with all entries equal to $0$ except for those whose label is in $\mc A$ that are equal to $1$. 
\tcr{The asymptotic notation $a\ll b$ and $a\sim b$ means $\lim a/b=0$ and $\lim a/b=1$, respectively.} 

\section{Averaging dynamics on general graphs}\label{sec:asymp}
Let $\mc G=(\mc V,\mc E,W)$ be a directed weighted graph representing the network, where $\mc V=\{1,\ldots,n\}$ is the set of nodes, $\mc E\subseteq \mc V\times\mc V$ is the set of links, and $W\in\R^{n\times n}$ is a matrix of nonnegative link weights such that $W_{ij}>0$ if and only if $(i,j)\in\mc E$, with positive diagonal elements of $W$ corresponding to self-loops. We refer to the graph $\mc G$ as: \emph{connected} if $W$ is irreducible;\tcr{\footnote{\tcr{Note that this convention deviates from the one adopted by some authors who refer to $\mc G$ as \emph{strongly} connected if $W$ is irreducible and simply connected if $W+W'$ is irreducible.}}}
\emph{undirected} if $W$ is symmetric; \emph{balanced} if $W\1=W'\1$; \emph{unweighted} if $W_{ij}\in\{0,1\}$ for all $i,j\in\mc V$.
We denote the out-degree vector by $w=W\1$ and assume\footnote{This assumption implies no loss of generality since one can add a self-loop with $W_{ii}>0$ to nodes $i$ with $w_i=0$ without modifying connectivity and other properties of $\mc G$.} that $w_i>0$ for all nodes $i$. We then introduce the matrices
\be\label{Pdef}D=\diag(w)\,,\quad P=D^{-1}W\,,\quad L=D-W
\,.\ee
Observe that the matrices $P$ and $-L$ are respectively row-stochastic and Metzler. Also, $P'w=w$ and $L'\1=0$ if and only if $\mc G$ is balanced. Moreover, $\mc G$ being undirected is equivalent to the detailed balance $w_iP_{ij}=w_jP_{ji}$ for $i,j\in\mc V$, a property that is referred to as reversibility of $P$ (with respect to $w$). The matrix $L$ is known as the graph \emph{Laplacian}. 

One of the most popular network dynamics can be seen as the interconnection of local averaging systems, 
i.e., multi-input/single-state dynamics placed at the nodes $i\in\mc V$ and governed by the linear updates 
$x_i(t+1)=\alpha x_i(t)+(1-\alpha)\sum_{j}P_{ij}u_j(t)$. 
Here, 
$\alpha\in[0,1]$ is an inertia parameter. 
By putting  $u_j(t)=x_j(t)$  one obtains the interconnected system
\beq\label{interconnected system} x_i(t+1)=\alpha x_i(t)+(1-\alpha)\sum_{j}P_{ij}x_j(t)\,,\eeq
for $i\in\mc V$.
In \eqref{interconnected system}, 
the sum index $j$ runs in principle over the whole node set $\mc V$, but is in fact restricted to the out-neighborhood $\mc N_i:=\{j:\,W_{ij}>0\}$ of node $i$ in $\mc G$. 
By assembling all the node states in a column vector $x(t)\in\R^n$, (\ref{interconnected system}) can be compactly rewritten as 
\beq\label{compact} x(t+1)=P_{\alpha} x(t)\,,\eeq
where $P_{\alpha}=\alpha I+(1-\alpha)P$. 
Hence, the state vector $x(t)$ of the \emph{distributed averaging} dynamics \eqref{compact} evolves as $x(t)=P_{\alpha}^tx(0)$, so that its asymptotic behavior is dictated by the eigen-structure of $P_{\alpha}$. 
Being a stochastic matrix, $P$ is non-expansive in the $||\cdot||_\infty$ norm, so that its spectrum is contained in the unitary disk centered in $0$. Hence, for $0\le\alpha\le1$, the matrix $P_{\alpha}$ has $1$ as eigenvalue (corresponding to right eigenvector $\1$) and its whole spectrum is contained in the closed disk of diameter coinciding with the segment joining the points $-1+2\alpha$ and $1$ in the complex plane. 
Finer properties of the spectrum of $P_\alpha$ are closely related to the geometrical properties of the graph $\mc G$ as summarized below. 

First we consider the case when the graph $\mc G$ is connected. In this case, it is a standard result of the Perron-Frobenius theory that $P_{\alpha}^t$ converges to a matrix $\1\pi'$ where $\pi$ can be uniquely characterized as the left eigenvector $\pi'=\pi'P$ such that $\1'\pi=1$. Connectivity of $\mc G$ implies that all the entries of $\pi$ ---which is referred to as the \emph{centrality} vector--- are strictly positive. For a balanced graph, $\pi$ is proportional to the degree vector, namely, $\pi=w/(\1'w)$. For general, unbalanced, connected graphs such simple expression does not hold true, while one can express the entries $\pi_i$ in terms of infinite sums.
For $\alpha\in[0,1)$, let the \emph{mixing time} of $P_{\alpha}$ be 
$$\tau_{\alpha}:=\inf\Big\{t\ge0:\,\max_{i\in\mc V}\sum\nolimits_{j}|(P_{\alpha}^t)_{ij}-\pi_j|\le \frac1{2e}\Big\}\,.$$
The mixing time is a popular index to study the speed of convergence of $P_{\alpha}^t$. In certain cases it can be estimated from knowledge of the second largest eigenvalue of $P_{\alpha}$ or coupling techniques. E.g., for the unweighted $d$-dimensional toroidal grid, one has $\tau_{\alpha}\sim C_dn^{2/d}$ where $C_d$ is a constant depending on the dimension $d$ but not on the graph size $n$. For general large-scale graphs whose spectrum analysis is unfeasible and for which no effective couplings are known, it proves more convenient to relate the mixing time to the graph \emph{conductance}
$$\Phi:=\min_{\emptyset\ne\mc U\subsetneq\mc V}\frac{\sum_{i\in\mc U}\sum_{j\in\mc V\setminus\mc U}\pi_iP_{ij}}{\sum_{i\in\mc U}\pi_i\cdot\sum_{j\in\mc V\setminus\mc U}\pi_j}\,,$$
that is a measure of the lack of bottlenecks in the graph. Results in \cite[Section 4.3]{MontenegroTetali} imply that 
\be\label{conductance-bound}\frac{1-2/e}{\Phi}\le \tau_{1/2}\le\frac1{\Phi^2}\log\frac{e^2}{\pi_*}\,,\ee
where $\pi_*=\min_{i\in\mc V}\pi_i$. By combining the bounds above with estimates of the conductance, it can be shown, e.g., that the Erdos-Renyi random graphs in the connected regime\footnote{They are constructed by considering $n$ nodes randomly putting a link between any pair of them independently with probability $p=c\log n/n$ for $c>1$.}  exhibit, with probability $1$, mixing times of the order of $\log n$. Such graphs are thus mixing faster then the $d$-dimensional tori.


%
%

These results imply that the state $x(t)$ of the averaging dynamics (\ref{compact}) on a connected graph converges to a consensus on a value $\ov x=\pi'x(0)$ that is the average of the nodes' initial values  weighted by their centralities. The speed of this convergence is captured by the mixing time as in the following.  

\begin{proposition}\label{prop convergence} Let $\mc G=(\mc V,\mc E,W)$ be a connected graph. Then, 
for every $\alpha\in(0,1)$, the distributed averaging dynamics \eqref{compact} satisfy 
\be
||x(t)-\1\ov x||_{\infty}\le||x(0)-\1\ov x||_{\infty}\exp(-\lfloor t/\tau_{\alpha}\rfloor)\,,\ee
for all $x(0)\in\R^n$, where
\be\ov x=\pi'x(0)\,,\ee
so that, in particular, 
$\lim\limits_{t\to +\infty}x(t)=\1\ov x\,.$
\end{proposition}
\tcr{\begin{proof}
As a consequence of the non-expansivity of $P'$ in $l_1$-distance we get that 
$$\sum\nolimits_{j}|(P_{\alpha})^t_{ij}-\pi_j|\le\exp(-\lfloor t/\tau_{\alpha}\rfloor)\,,\qquad i\in\mc V\,.$$
(Cf., e.g., \cite[Eqn.~(4.34)]{LevinPeresWilmer}.) Then, for all $i\in\mc V$,  
$$\ba{rcl}|x_{i}(t)-\ov x|&=&|\sum_j((P_{\alpha}^t)_{ij}-\pi_j)(x_j(0)-\ov x)|\\
&\le&|\sum_j(P_{\alpha}^t)_{ij}-\pi_j|\cdot|x_j(0)-\ov x|\\ 
&\le&\exp(-\lfloor t/\tau_{\alpha}\rfloor)\cdot||x(0)-\ov x||_{\infty}\,,\ea$$
which gives the claim. \qed\end{proof}
}


%

We now move on to discussing the asymptotic behavior of the distributed averaging dynamics \eqref{compact} in arbitrary ---not necessarily connected--- graphs. 
The node set $\mc V$ can always be uniquely partitioned as $\mc V=\mc V_1\cup\cdots \cup \mc V_{c}$ where, for $1\le k\le c$, the subgraph $\mc G_k=(\mc V_k, \mc E\cap (\mc V_k\times\mc V_k), W_{|\mc V_k\times\mc V_k})$ is connected 
and maximal with respect to this property. Such subgraphs are called \emph{connected components}. For any two of them, $\mc G_h$ and $\mc G_k$, we write $\mc G_h\ge\mc G_k$ if there is a path in $\mc G$ connecting some node in $\mc V_k$ to some node in $\mc V_h$. By construction the relation is transitive and such that $\mc G_h\ge\mc G_k$ and $\mc G_k\ge\mc G_h$ if and only if $\mc G_k=\mc G_h$. Connected components which are maximal with respect to the partial ordering $\ge$ are called \emph{sink components}: any path starting in a sink component will never leave it.
The following result characterizes the asymptotic behavior of the averaging dynamics \eqref{compact} on a general graph. 

\begin{theorem}\label{theo convergence} Let $\mc G=(\mc V,\mc E,W)$ be a graph with sink components $\mc G_k=(\mc S_k, \mc E\cap (\mc S_k\times\mc S_k), W_{|\mc S_k\times\mc S_k})$,  for $k=1,\dots , s$. Let $\mc S:=\bigcup_{1\le k\le s}\mc S_k$ and $\mc S_{-k}:=\mc S\setminus\mc S_k$. Then:
\begin{enumerate}
\item[(i)] There exists a unique nonnegative $H\in\R^{n\times s}$ such that $H\1=\1$ and, for $i\in\mc S$, $1\le k\le s$, 
\be\label{H-characterization}LH=0\,,\qquad H_{ik}=\left\{\ba{ll}1&i\in\mc S_k\\0& i\in\mc S_{-k}\,;\ea\r.
\ee
\item[(ii)] For every $\alpha\in(0,1)$, the distributed averaging dynamics \eqref{compact} satisfies
\be\label{limx}\lim\limits_{t\to +\infty}x(t)=H\ov x\,,\ee
for all $x(0)\in\R^n$, where $\ov x\in\R^s$ has entries
\be\label{ovxj} \ov x_k=\sum_{i\in\mc S_k}\pi^{(k)}_ix_i(0)\,,\qquad k=1,\dots,s\,,\ee
and $\pi^{(k)}\in\R^{\mc S_k}$ is the centrality vector of $\mc G_k$.
\end{enumerate}
\end{theorem}

\begin{proof}
Upon reordering nodes in such a way that sink components come last, $P_{\alpha}$ takes the form
$$P_{\alpha}=\left(\begin{matrix} Q&R^{(1)}&\dots &\ldots&R^{(s)}\\
0&P^{(1)}&0&\ldots&0\\
0&0&P^{(2)}&&0\\
\vdots &\vdots&\vdots&\ddots&\vdots \\
0&0&0&\ldots&P^{(s)}\end{matrix}\right)_{\ds \,.}$$
Splitting $x(t) =(y(t),x^{(1)}(t),\dots , x^{(s)}(t))$ accordingly, with $y(t)\in\R^{\mc R}$ where $\mc R:=\mc V\setminus\mc S$, and  $x^{(k)}(t)\in \R^{\mc S_k}$ for $k=1,\dots ,s$, recursion \eqref{compact} reads
\beq\label{splitting}\begin{array}{rcl}y(t+1)&=&\ds Qy(t)+\sum_{1\le k\le s}R^{(k)}x^{(k)}(t)\\[15pt] x^{(k)}(t+1)&=&P^{(k)}x^{(k)}(t)\,,\quad1\le k\le s\,.\end{array}\eeq
First, notice that the evolution of the state on the nodes of the sink components can be studied using Proposition \ref{prop convergence}. In particular, for $k=1,\dots ,s$, we have that $x^{(k)}(t)\stackrel{t\to\infty}{\longrightarrow}\1\ov x_k$, where $\ov x_k$ is as in \eqref{ovxj} and $\pi^{(k)}\in\R^{\mc S_k}$ is the centrality vector of $\mc G_k$. 

On the other hand, $Q$ is nonnegative, 
and the Perron-Frobenius theory implies that its spectral radius $\rho$ is an eigenvalue with associated nonnegative {left} eigenvector $z\in\R^{\mc R}$. Let $\mc J\subseteq\mc R$ be the support of $z$. Since every node in $\mc R$ is connected to some node in $\mc S$, $\min_{j\in\mc J}\sum_{i\in\mc J}Q_{ji}<1$ (otherwise there would be no links from $\mc J$ to $\mc V\setminus\mc J$). Hence,  $$\rho\sum_{i\in\mc J}z_i=\sum_{i\in\mc J}\sum_{j\in\mc R}Q_{ji}z_j=\sum_{i\in\mc J}\sum_{j\in\mc J}Q_{ji}z_j<\sum_{j\in\mc J}z_j\,,$$ i.e., $\rho<1$. 
\tcr{Thus, $Q$ is stable and $I-Q$ is invertible. 
On the other hand, $\lim_{t\to\infty}x^{(k)}(t)=\1\ov x_k$ for $1\le k\le s$ by Proposition \ref{prop convergence}. 
It follows that the first line of \eqref{splitting} is a stable LTI system with converging input, so that its state is necessarily converging to
$$\lim_{t\to+\infty}y(t)=(I-Q)^{-1}\sum_{1\le k\le s}R^{(k)}\1\ov x_k\,.$$}
This yields \eqref{limx}, with $H\in\R^{n\times s}$ defined by
\beq\label{H}H_{ik}=\left\{  \begin{array}{ll}((I-Q)^{-1}R^{(k)}\1)_i\quad &\, i\in\mc R\\ 1\quad &\, i\in\mc S_k\\ 0\quad &\, i\in\mc S_{-k}\,.\end{array}\right.\eeq
Note that $H$ is nonnegative since both $R^{(k)}$, for $1\le k\le s$, and $(I-Q)^{-1}=\sum_{l\ge0}Q^l$ are.
Moreover, $P_{\alpha}\1=\1$ implies $\sum_{1\le k\le s}R^{(k)}\1=(I-Q)\1$, so that
$$\sum_{1\le k\le s}H_{ik}=((I-Q)^{-1}\sum_{1\le k\le s}R^{(k)}\1)_i=1\,,\quad i\in\mc R\,.$$
Hence, $H\1=\1$. Furthermore, one has that 
$$\ba{rcl}\!\!\!LH\!\!&\!\!\!=\!\!\!&
\ds D(I-\tcb{P})H=\frac1{1-\alpha}D(I-P_{\alpha})\\[5pt]
&\!\!\!=\!\!\!&\ds\frac1{1-\alpha}D(I-Q)(I-Q)^{-1}\!\!
\summ_{1\le k\le s}R^{(k)}\1\1_{\{k\}}'\\[5pt]&&
\ds-\frac1{1-\alpha}D\summ_{1\le k\le s}R^{(k)}\1\1_{\{k\}}'=0\,,\ea$$
proving \eqref{H-characterization}. Uniqueness of the solution of \eqref{H-characterization} follows from invertibility of the $\mc R\times\mc R$ block of $L$.
\qed\end{proof}

Theorem \ref{theo convergence} states that the nodes belonging to a sink component $\mc S_k$ asymptotically reach consensus on the value  $\ov x_k=\sum_{i\in\mc S_k}\pi^{(k)}_ix_i(0)$. The state of every other node $i\in\mc R$ converges to a convex combination of the consensus values of the various sink components with weights $H_{ik}$ characterized by \eqref{H-characterization}. 
The initial states $x_i(0)$ of the nodes $i\in\mc R$ have thus no influence on the equilibrium state $x$. 
Equivalently, the equilibrium state $x$ of the averaging dynamics \eqref{compact} can be characterized as the solution of 
\be\label{Laplace+boundary}Lx=0\,,\qquad x_j=\ov x_k\qquad j\in\mc S_k\,,\ 1\le k\le s\,,\ee
which is refereed to as the \emph{Laplace equation} on $\mc G$ with \emph{boundary conditions} on $\mc S=\mc S_1\cup\ldots\cup\mc S_k$.

Two special cases are worth being examined:
\begin{itemize}
\item If there is a single sink component, $H=\1$, thus Theorem \ref{theo convergence} implies that 
$P_{\alpha}^tx(0)\stackrel{t\to\infty}{\longrightarrow}\1\pi'x(0)$. I.e., 
the system converges to consensus on a convex combination of the initial states of the nodes of the sink component, while all other nodes' initial states do not play any role. This is a generalization of Proposition \ref{prop convergence}.
\item If the sink components are all singletons, i.e., $\mc S_k=\{v_k\}$ for $1\le k\le s$, then
$x(t)\stackrel{t\to\infty}{\longrightarrow}H\ov x$, where $\ov x_k=x_{v_k}(0)$ for $1\le k\le s$. 
In this case, sink nodes keep their state constant in time. They are sometimes referred to as \emph{stubborn nodes} and are used to model opinion leaders in social networks \cite{Acemoglu.etal:2013}, or anchor nodes in robotic formation models \cite{Ji.ea:08}.
\end{itemize}

An interesting probabilistic interpretation comes from considering a discrete-time Markov chain $X(t)$ with state space $\mc V$ and transition probability matrix $P_{\alpha}$. I.e., $X(t+1)$ is conditionally independent from the past $X(0),\ldots,X(t-1)$ given the present $X(t)$, and $\P(X(t+1)=j|X(t)=i)=(P_{\alpha})_{ij}$. It is well known that, with probability one,  $X(t)$ will enter one of the sink components in finite time and never leave it ever after. For $k=1,\dots ,s$, let $A_{k}$ be the event that $X(t)$ enters the sink component $\mc S_k$ before any other. Consider the matrix $M\in\R^{n\times s}$ with entries $M_{ik}=\P(A_{k}\,|\, X(0)=i)$. Then, a simple conditioning argument yields
$$\ba{rcl} M_{ik}&=&\P(A_{k}\,|\, X(0)=i)\\[10pt]&=&\sum\limits_{l\in\mc V}(P_{\alpha})_{il}\P(A_{k}\,|\, X(1)=l)=(P_{\alpha}M)_{ik}\ea$$
for all $i\in\mc V$ and $1\le k\le s$. Hence, $M$ solves \eqref{H-characterization} and we can deduce that $H=M$. In other terms, the  weight $H_{ik}$ that agent $i$ puts on $\ov x_k$ in determining its asymptotic state $x_i$ in \eqref{compact} can be interpreted as the probability that a Markov chain started at node $i$ and moving with transition probability matrix $P_{\alpha}$ hits the sink component $\mc S_k$ before any other.


\section{The electrical network interpretation}\label{sec:electrical}
Consider a graph $\mc G=(\mc V,\mc E,W)$ with sink components $\mc S_1,\ldots ,\mc S_s$, where $s\ge2$. Put $\mc S:=\bigcup_{1\le k\le s}\mc S_k$, $\mc R:=\mc V\setminus\mc S$, and $\mc S_{-k}:=\mc S\setminus\mc S_k$ for every $1\le k\le s$. 
Then, Theorem \ref{theo convergence} guarantees that the state of the averaging dynamics \eqref{compact} converges to an equilibrium $x=H\ov x$, where $H\in\R^{n\times s}$ is the stochastic matrix satisfying \eqref{H-characterization}
and $\ov x\in\R^s$ is the vector of the weighted averages of the initial condition in the sink components. 
In this section, we focus on the special case when the restriction of the graph $\mc G$ to the node set $\mc R$ is connected and undirected, i.e., when 
\be\label{symmetryW}W_{ij}=W_{ji}\,,\qquad i,j\in\mc R\,.\ee

We will first interpret the equilibrium states of \eqref{compact} as voltages in an electrical network associated to $\mc G$, then relate them to the effective resistances in the network.  
Let us start by defining the link flows
\be\label{Ohm}f_{ij}=W_{ij}(x_i-x_j)\,,\quad i,j\in\mc V\,.\ee
The key consequence of \eqref{symmetryW} is that then $x=H\ov x$ and \eqref{H-characterization} imply the following conservation law: 
\be\label{Kirchoff}
\sum_{j}f_{ij}
=0\,,\qquad\forall i\in\mc R\,.
\ee
Indeed, one can give the following interpretation: the link weights $W_{ij}$ represent conductances and their inverses are resistances; $x_i$ is the voltage in node $i$; and $f_{ij}$ is the electrical current flowing from node $i$ to node $j$. 
Then, \eqref{Ohm} and \eqref{Kirchoff} can be read as the Ohm law and, respectively, the Kirchoff law. 

Such interpretation has deep implications. First, by simply verifying first-order conditions, one can show that \eqref{symmetryW} implies that, for $k=1,\ldots,s$, the $k$-th column of $H$ coincides with the solution of the following quadratic optimization problem
\be\label{heat-dissipation-1}\frac1{R_{\mc S_k\leftrightarrow S_{-k}}}=\min_{\substack{\ds y\in\R^n:\\ \ds y_i=1\ \quad i\in\mc S_k\\\ds y_i=0\ \ i\in\mc S_{-k}}} \frac12\sum_{i,j\in\mc V}W_{ij}(y_i-y_j)^2\,.\ee
The quantity that is to be minimized in the righthand side of \eqref{heat-dissipation-1} represents the energy dissipation in the network when the voltages  are $y_i$. Hence, $H\1_{\{k\}}$ is the vector of voltages with minimal energy dissipation under the constraints that the voltage is $1$ in $\mc S_k$ and $0$ in $\mc S\setminus\mc S_k$. The inverse $R_{\mc S_k\leftrightarrow S_{-k}}$ of such minimal energy dissipation  is known as the \emph{effective resistance} between the node sets $\mc S_k$ and $\mc S_{-k}$. \tcr{A classical duality result known as Thompson's principle \cite[Th.~9.10]{LevinPeresWilmer} states that $R_{\mc S_k\leftrightarrow S_{-k}}$ coincides with the minimal energy dissipation of a unitary flow from $\mc S_k$ to $\mc S_{-k}$ 
\be\label{heat-dissipation-2}{R_{\mc S_k\leftrightarrow\mathcal S_{-k}}}=\min_{\substack{\ds\theta\in\R^{n\times n}:\\
\ds(\theta\1)'\1_{\mc S_{k}}=1\\
\ds(\theta\1)_i=0\quad i\in\mc R
}} \frac12\sum_{i,j\in\mc V}\frac1{W_{ij}}\theta_{ij}^2\,,\ee
and that the minimum above is achieved by \be\label{Thompson1}\theta_{ij}=W_{ij}(H_{ik}-H_{jk})R_{\mc S_k\leftrightarrow\mc S_{-k}}\,.\ee 
{In particular, the fact that $\theta$ as defined in \eqref{Thompson1} is a unitary flow from $\mc S_k$ to $\mc S_{-k}$ (i.e., it satisfies the constraints in the righthand side of \eqref{heat-dissipation-2}) implies that its normalized version $R_{\mc S_k\leftrightarrow\mc S_{-k}}^{-1}\theta$  satisfies the following flow conservation equations} 
\be\label{Thompson}
\ba{rclc}\ds\frac1{R_{\mc S_k\leftrightarrow S_{-k}}}
&\!\!\!\!=\!\!\!\!&\ds\summ_{i\in\mc V}\summ_{j\in\mc S_k}W_{ij}(1-H_{ik}) &\left(\!\!\!\!\ba{c}\text{flow}\\\text{out of }\mc S_{k}\ea\!\!\!\!\right)\\
&\!\!\!\!=\!\!\!\!&\ds\summ_{i\in\mc U}\summ_{j\in\mc W}W_{ij}(H_{ik}-H_{jk}) &\left(\!\!\ba{c}\text{flow}\\\text{from }\mc U\\\text{ to }\mc W\ea\!\!\right)\\[15pt]
&\!\!\!\!=\!\!\!\!&\ds \sum_{i\in\mc V}\sum_{j\in\mc S_{-k}}W_{ij}H_{ik}\,,&\left(\!\!\!\!\ba{c}\text{flow}\\\text{into }\mc S_{-k}\ea\!\!\!\!\right)
\ea
\ee
where $\mc V=\mc U\cup\mc W$, $\mc U\cap\mc W=\emptyset$ is any cut of $\mc G$ such that $\mc S_{k}\subseteq\mc U$ and $\mc S_{-k}\subseteq\mc W$.}

One key advantage of the electrical network interpretation is that the relations \eqref{heat-dissipation-1} and \eqref{Thompson} imply simple and powerful rules, e.g., the parallel and series laws \cite[pp.~135--136]{LevinPeresWilmer}, to compute or estimate the equilibrium state $x$. 
Another fundamental property is monotonicity with respect to changes of the network (known as Rayleigh's law): the effective resistance is never increased when new links are added, or when the conductance, i.e., the weight, of some existing links is increased ---including when two nodes in $\mc R$ are glued together, that is equivalent to the addition of an infinite weight link between them. \cite[Theorem 9.12]{LevinPeresWilmer}
\tcr{In fact, the equilibrium states $x_i$ can be expressed purely in terms of effective resistances as stated in the following.  
\begin{theorem}\label{theorem-resistances} Let $\mc G=(\mc V,\mc E,W)$ be a graph with sink components $\mc S_1,\ldots ,\mc S_s$, $s\ge2$, whose restriction to $\mc R=\mc V\setminus\mc S$, where $\mc S=\bigcup_{1\le k\le s}\mc S_k$, is connected and undirected (as per \eqref{symmetryW}). Then, the equilibrium state vector $x$ of \eqref{compact} satisfies, for $i\in\mc V$,
\beq\label{electrical bias}x_i=\frac12\sum_{1\le k\le s}\ov x_k\left(1+\frac{R_{i\leftrightarrow\mc S_{-k}}-R_{i\leftrightarrow\mc S_{k}}}{R_{\mc S_k\leftrightarrow\mc S_{-k}}}\right)\,,\eeq 
where $R_{i\leftrightarrow\mc S_{k}}$ and $R_{i\leftrightarrow\mc S_{-k}}$ are the effective resistances between node $i$ and $\mc S_k$ and $\mc S_{-k}$, respectively, in the graph obtained  from $\mc G$ by glueing together all nodes in $\mc S_k$ and $\mc S_{-k}$ respectively into single nodes.
\end{theorem}
{\begin{proof}
Let $\hat{\mc G}=(\hat{\mc V},\hat{\mc E},\hat W)$ be the connected undirected graph obtained from $\mc G$ by merging all nodes in $\mc S_k$ and $\mc S_{-k}$ into single nodes $v$ and $\ov v$, respectively, and making the links incident in $v$ and $\ov v$ bidirectional. Let its Laplacian $\hat L=\diag(\hat W\1)-\hat W$ have eigenvalues $0=\lambda_1<\lambda_2\le\ldots\le\lambda_n$ and corresponding orthonormal base of eigenvectors $\frac1{\sqrt n}\1=\phi_{(1)},\phi_{(2)},\ldots,\phi_{(n)}$. Define the Green matrix 
$$G=G'=\sum_{2\le l\le n}\frac1{\lambda_l}\phi_{(l)}\phi_{(l)}'$$ 
and observe that $G\1=0$ and
$$LG=\sum_{2\le l\le n}\frac1{\lambda_l}L\phi_{(l)}\phi_{(l)}'=\sum_{2\le l\le n}\phi_{(l)}\phi_{(l)}'=I-\frac1n\1\1'\,.$$
From the above and $L\1=0$ it can be deduced that, for $h,j\in\hat{\mc V}$, all solutions $y$ of $Ly=(\1_{\{h\}}-\1_{\{j\}})$ can be written as $y=G(\1_{\{h\}}-\1_{\{j\}})+\alpha\1$ for some scalar $\alpha$. 
Now, let $\hat R_{hj}$ be the effective resistance between $h$ and $j$ in $\hat{\mc G}$
and $x^{(h,j)}$ be the solution of  $(Lx^{(h,j)})_i=0$ for \tcb{$i\ne h,j$}, 
$x^{(h,j)}_h=1$, $x^{(h,j)}_j=0$.  
Arguing as in \eqref{Thompson} gives $Lx^{(h,j)}=\hat R_{hj}^{-1}(\1_{\{h\}}-\1_{\{j\}})$,
from which we deduce that 
\be\label{R1}\hat R_{hj}x^{(h,j)}=G(\1_{\{h\}}-\1_{\{j\}})+\alpha\1\,,\ee 
for some scalar $\alpha$.  
It follows from \eqref{R1} that 
\be\label{R2}\ba{rcl}\hat R_{hj}
&=&\hat R_{hj}(x^{(h,j)}_h-x^{(h,j)}_j)\\
&=&(\1_{\{h\}}-\1_{\{j\}})'G(\1_{\{h\}}-\1_{\{j\}})\\
&=&G_{hh}-G_{hj}-G_{jh}+G_{jj}\\
&=&G_{hh}-2G_{hj}+G_{jj}\,.
\ea\ee
By applying \eqref{R1} with $h=v$ and $j=\ov v$, and \eqref{R2} twice, with $h=l$ and $j=v$ first, and then with $h=l$ and $j=\ov v$, one gets, for $l\in\hat{\mc V}$, 
\be\label{R3}\!\!\!\!\!\ba{rcl}
2\hat R_{v\ov v}   x_l^{(v,\ov v)}\!\!
&\!\!=\!\!&2(G_{lv}-G_{l\ov v}+\alpha)\\
&\!\!=\!\!&{G_{vv}-\hat R_{lv}-G_{\ov v\ov v}+\hat R_{l \ov v}}+2\alpha\,.
\ea\ee
Choosing $l=\ov v$ and recalling that $x_{\ov v}^{(v,\ov v)}=0=\hat R_{\ov v\ov v}$  gives \tcb{$2\alpha=\hat R_{v\ov v}+G_{\ov v\ov v}-G_{vv}$} in \eqref{R3}. Substituting back in \eqref{R3} and noting that $x_i^{(v,\ov v)}=H_{ik}$, one gets 
\be\label{R}H_{ik}=\frac{R_{\mc S_k\leftrightarrow\mc S_{-k}}+R_{i\leftrightarrow\mc S_{-k}}-R_{i\leftrightarrow\mc S_{k}}}{2R_{\mc S_k\leftrightarrow\mc S_{-k}}}\,,\ee
for $i\in\mc V$ and  $1\le k\le s$. The claim now follows by substituting \eqref{R} into \eqref{limx}. 
We observe that alternative proofs of \eqref{R} have been proposed based on Markov chain arguments (cf.~\cite[Ex.~10.8]{LevinPeresWilmer}). 
\qed\end{proof}}

An insightful special case of Theorem \ref{theorem-resistances} is when $\mc G$ has two sink components $\mc S^+$ and $\mc S^-$ with values $\ov x^+=1$ and $\ov x^-=0$, respectively. Then, \eqref{electrical bias} reads
$$x_i=\frac12+\frac{R_{i\leftrightarrow\mc S^{-}}-R_{i\leftrightarrow\mc S^{+}}}{2R_{\mc S^+\leftrightarrow\mc S^{-}}}\,,\qquad i\in\mc V\,.$$
The sign of the difference between the two effective resistances determines if node $i$ will be more influenced by $\mc S^+$ or $\mc S^-$. In this sense, formula (\ref{electrical bias}) expresses the bias of a node towards a sink component as determined by the electrical resistance to that sink in comparison to the others. Such electrical equivalence has recently found applications in the design of efficient distributed algorithms for the optimal stubborn node placement problem. \cite{Vassio.etal:2014}


\section{Polarization and homogeneous influence} \label{sec:polarizedvshomo}
In this section, we consider graphs 
with $s\ge2$ singleton sink components $\mc S_1=\{v_1\},\ldots,\mc S_s=\{v_s\}$ with values $\ov x_1,\ldots, \ov x_s\in[0,1]$. We investigate conditions  ---on the graph structure and on the size of the sink components--- under which most of the entries $x_i$ of the equilibrium state $x=H\ov x$ of the averaging dynamics \eqref{compact} are close to a common value $\tilde x$ which is a convex combination of the $\ov x_j$s, or rather they are all close to one of the extreme values $\ov x_1,\ldots, \ov x_s$. In order to formalize these notions, we consider infinite sequences of graphs (typically of increasing size), and briefly refer to them as (large-scale) networks. 
Following \cite{Acemoglu.etal:2013}, we say that the sink components $\mc S=\{v_1,\ldots,v_s\}$ have \emph{homogeneous influence} on (the rest of) the network if,  for all $\eps>0$,
\be\label{homogeneous-influence} \lim\limits_{n\to\infty}\inf_{\tilde x}\frac1n\left|\left\{i\in\mc V:\,|x_i-\tilde x|<\eps\right\}\right| =1\,.\ee
On the other hand, we refer to a network as \emph{polarized} if, for all $\eps>0$,
\be\label{polarization} \lim\limits_{n\to\infty}\frac1n\left|\left\{i\in\mc V:\,\min_{1\le k\le s}|x_i-\ov x_k|<\eps\right\}\right| =1\,.\ee
As in the previous section, we confine our discussion to the special case when the restriction of the graph $\mc G$ to $\mc R=\mc V\setminus\mc S$ is connected and undirected.
\tcr{In fact, for the sake of simplicity, we mostly focus on an even more special graph structure} 
$\mc G=(\mc V,\mc E,W)$ whose node set $\mc V$ consists of only two stubborn nodes, $v_0$ and $v_1$, and two disjoint sets of regular nodes, $\mc U_0$ and $\mc U_1$, such that: the nodes in $\mc U_0$ (respectively, $\mc U_1$) are all connected by a weight-$\gamma$ directed link to $v_0$ ($v_1$); the subnetwork obtained by removing the stubborn nodes from $\mc G$ is undirected and connected; 
the aggregate weight $\sum_{j\in\mc U_1}W_{ij}$ (respectively, $\sum_{j\in\mc U_0}W_{ij}$) of links connecting a node $i\in \mc U_0$ ($i\in \mc U_1$) to nodes in $\mc U_1$ ($\mc U_0$) is a positive constant $\beta_0$ ($\beta_1$) independent of $i$. In other words, we consider a network whose weight matrix $W$ has the structure 
\be\label{Wstructure}W=\left[\!\!\!\ba{cccc}\gamma&\!\!\!\!0\ldots&\ldots 0\!\!\!\!&0\\ 
\ba{c}\vdots\\\gamma\ea&A&B&\ba{c}\vdots\\0\ea\\
\ba{c}0\\\vdots\ea&C&D&\ba{c}\gamma\\\vdots\ea\\  
0&\!\!\!\!0\ldots&\ldots 0\!\!\!\!&\gamma\ea\!\!\!\right]\quad 
\ba{ccc}
 A=A'\\[3pt] B=C'\\[3pt] C=B'\\[3pt] D=D'\\[3pt] 
 B\1=\beta_0\1\\[3pt]C\1=\beta_1\1\,.\ea
\ee
\begin{proposition} \label{proposition:polarization}
Let $\mc G=(\mc V,\mc E,W)$ be a graph with  weight matrix $W$ as in \eqref{Wstructure}. Let the stubborn nodes be assigned values $\ov x_{v_0}=0$ and $\ov x_{v_1}=1$. 
Let $x=H\ov x$ be the equilibrium state of the averaging dynamics \eqref{compact} and  $y_h:=\frac1{n_h}\sum_{i\in\mc U_h}x_i$ for $h=0,1$ be the average states in the two subsets of nodes. 
Then, 
$$\ba{c} \ds |h-y_h|\le\l({\ds 1+\frac{n_{h}}{n_{1-h}}+\frac{\gamma}{\beta_h}}\r)^{-1}\,,\qquad h=0,1\,,\\[5pt]
\ds y_1-y_0\le\l({\ds 1+\frac{\beta_0}{\gamma}+\frac{\beta_1}{\gamma}}\r)^{-1}
\,.\ea$$
\end{proposition}
\begin{proof} 
Since the restriction of the graph to $\mc R$ is undirected, we can use the electrical circuit interpretation of Section \ref{sec:electrical}. In particular, \eqref{Thompson} yields
\be\label{conservation}
\ba{rcl}\ds\frac{1}{R_{v_0 \leftrightarrow v_1}}&=&
{\ds\gamma\sum_{i\in\mc U_0}x_i} 
={\ds\sum_{i\in\mc U_0}\sum_{j\in\mc U_1}W_{ij}(x_j-x_1)}\\[15pt]
&=&{\ds\gamma\sum_{j\in\mc U_1}(1-x_j)}\,.\ea
\ee
Moreover, we can get a lower bound on the effective resistance by merging, for $h=0,1$, all nodes in $\mc U_h$ into a single node $u_h$ and applying the parallel and series law to the resulting network, thus getting 
$$R_{v_0\leftrightarrow v_1}\ge\frac1{\gamma n_0}+\frac1{n_0\beta_0}+\frac1{\gamma n_1}=\frac1{\gamma n_0}+\frac1{n_1\beta_1}+\frac1{\gamma n_1}\,.$$
Then, the claim follows by substituting the identity \eqref{conservation} in the lefthand side of the above. 
%
\qed\end{proof}

It follows from Proposition \ref{proposition:polarization} that, for a network structure as in \eqref{Wstructure}, 
\begin{description}
\item[(i)] if $\gamma\gg\max\{\beta_0,\beta_1\}$, then $y_0\to0$ and $y_1\to1$; 
\item[(ii)] if $\gamma\ll\max\{\beta_0,\beta_1\}$, then $y_1-y_0\to0$. 
\end{description}
Observe that, since, for all $\eps>0$ and $h=0,1$,
$$\frac1n|\{i\in\mc U_h:\,|x_i-h|\ge\eps\}|\le\frac{|h-y_h|}{\eps}\,,$$
point (i) above implies that the network is polarized if $\gamma\gg\max\{\beta_0,\beta_1\}$. The intuition behind this result is that, if $\gamma\gg\max\{\beta_0,\beta_1\}$, then, in both $\mc U_0$ and $\mc U_1$, the total weight of links towards the stubborn node $v_0$ (respectively, $v_1$) is much larger than the total weight of links to the other set of nodes. Then, for $h=0,1$, nodes in $\mc U_h$ get influenced by $v_h$ much more than by the other set of nodes $\mc U_{1-h}$. 


On the other hand, observe that $y_1-y_0\to0$ does not imply homogeneous influence. 
Sufficient conditions for homogeneous influence have been proved in \cite{Acemoglu.etal:2013} based on finer properties of the graph $\mc G$, in particular on its mixing time, as per the following. 

\begin{theorem}[\cite{Acemoglu.etal:2013}]\label{theo:homog}
Let $\mc G=(\mc V,\mc E,W)$ be a graph with sink components $\mc S_1,\ldots,S_s$, let $\mc S=\bigcup_{1\le j\le s}\mc S_j$ and $\mc R=\mc V\setminus\mc S$. Let $\tilde{\mc G}=(\mc V,\tilde{\mc E},\tilde W)$ be the undirected graph obtained from $\mc G$ by 
making all the directed links from some node in $\mc R$ to some node in $\mc S$ bidirectional 
and letting the modified weight matrix $\tilde W$ coincide with $W$ in its $\mc R\times \mc V$ block, and be such that 
$\tilde W_{ij}=W_{ji}$ for all $i\in\mc S$ and $j\in\mc R$, and $\tilde W_{ij}=0$ for all $i,j\in\mc S$. 
Assume that $\tilde{\mc G}$ is connected, and let $\tilde P=\diag(\tilde W\1)^{-1}\tilde W$, $\tilde\pi=\tilde P'\tilde\pi$ be its invariant probability vector, and $\tilde\tau$ be the mixing time of $\frac12(I+\tilde P)$. 
Then, the equilibrium state $x$ of  \eqref{compact}  satisfies 
$$\frac1n\left|\left\{i\in\mc V:\, |x_i-\tilde x|\ge\eps\right\}\right|\le\frac{\Delta}{\eps n\tilde\pi_*}\psi(\tilde\tau\cdot\tilde\pi_{\mc S})\,,$$
for $\eps>0$, where $\tilde x=\tilde\pi'x$, $\Delta=\max_{1\le i,j\le s}\{\ov x_i-\ov x_j\}$, $\tilde\pi_*:=\min_{i\in\mc V}\tilde\pi_i$,  and $\psi(y):=y\log(e^2/y)$. 
\end{theorem}

Theorem \ref{theo:homog} implies that influence is homogenous in networks such that the product $\tilde\tau\cdot\tilde\pi_{\mc S}$ of the mixing time and the aggregate centrality of the set of stubborn nodes is vanishing (and $n\tilde\pi_*$ is bounded away from $0$). 
Networks such that $\tilde\tau\cdot\tilde\pi_{\mc S}\to0$ have been referred to as \emph{highly fluid} \cite{Acemoglu.etal:2013}. Examples of highly fluid networks include $d$-dimensional tori with $d\ge3$ when $|\mc S|\ll n^{1-2/d}$, and expansive networks such as the Erdos-Renyi graph provided that $|\mc S|\ll n/\log n$.

\begin{figure}
\begin{center}
\includegraphics[height=3.5cm]{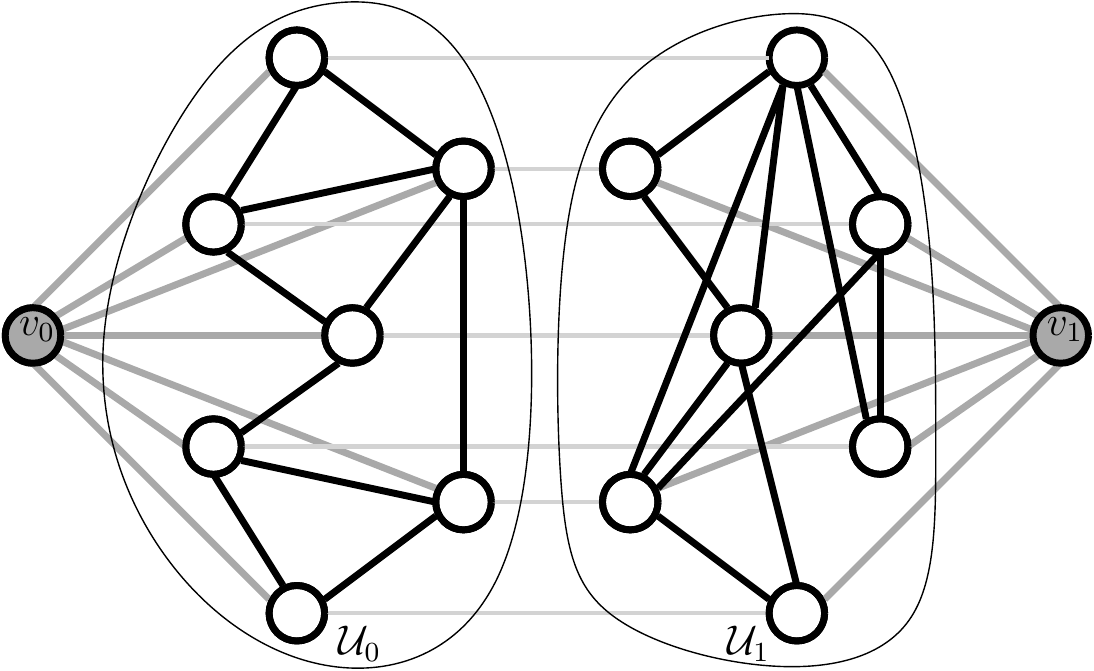}
\end{center}
\caption{\tcb{\label{fig:G} A graph $\mc G$ obtained by matching two independent and identically distributed Erdos-Renyi graphs (black links) with node sets $\mc U_0$ and $\mc U_1$ by weight-$\beta$ links (light grey horizontal links) and connecting, for $h=0,1$, each node in $\mc U_h$ with a stubborn node $v_h$ by a weight-$\gamma$ link (dark grey).}}
\end{figure}

We conclude this section with an application of Proposition \ref{proposition:polarization} and Theorem \ref{theo:homog}, highlighting a threshold phenomenon
, with a transition from polarization to homogeneous influence. Let $\mc G_0=(\mc U_0,\mc E_0)$ and $\mc G_1=(\mc U_1,\mc E_1)$ be two independent and identically distributed Erdos-Renyi random graphs  with parameters $|\mc U_0|=|\mc U_1|=m$ and $p=\omega m^{-1}\log m$, where $\omega>1$ is a constant independent of $m$. 
Hence, distinct pairs of nodes $\{i,j\}\subseteq\mc U_{h}$, $h=0,1$, are connected by weight-$1$ undirected links independently with probability $p$. The scaling $pm/\log m=\omega>1$ guarantees that $\mc G_0$ and $\mc G_1$ are connected with high probability as $m$ grows large. \cite[Ch.~6]{Durrettbook}
Then, let $\mc G=(\mc V,\mc E,W)$, where $\mc V=\{v_0\}\cup\mc U_0\cup\mc U_1\cup\{\tcr{
v_1}\}$  be the graph obtained by interconnecting $\mc G_0$ and $\mc G_1$ by an arbitrary matching of $\mc U_0$ and $\mc U_1$ of weight-$\beta$ links (i.e., every node in $\mc U_h$ is connected to exactly one node in $\mc U_{1-h}$ by an undirected weight-$\beta$ link) and adding a directed weight-$\gamma$ link from each node in $\mc U_0$ to $v_0$ and from each node in $\mc U_1$ to $v_1$. (See Figure \ref{fig:G}.)
Proposition \ref{proposition:polarization} and Theorem \ref{theo:homog} imply that 
\begin{description}
\item[(i)] if $\gamma\gg\beta$, then the network is polarized; 
\item[(ii)] if $\gamma\ll\beta\ll1$, then influence is homogeneous.
\end{description}
Indeed, point (i) above follows directly from Proposition \ref{proposition:polarization}. 
In order to verify point (ii), for $h=0,1$, consider the network $\hat{\mc G}_h$ with regular nodes $\mc U_h$ and stubborn nodes $\mc S_h:=\{v_h\}\cup\mc U_{1-h}$ obtained by removing from $\mc G$ node $v_{1-h}$ along with all the internal links of $\mc G_{1-h}$. 
\tcr{Let $\tilde{\mc G}_h$ be the undirected graph obtained from $\hat{\mc G}_h$ by making all its links incident to $v_h$ and to any $v\in\mc U_{1-h}$ bidirectional with weight $\gamma$  and $\beta$, respectively. (See Figure \ref{fig:G0}.)}
Let $l$ be the total number of undirected links in the Erdos-Renyi graph $\mc G_h$, that is of order $m^2p=\omega m\log m$ with high probability. Note that, in $\tilde{\mc G}_h$, \tcb{the degree of $v_h$ is $m\gamma$ ($m$ weight-$\gamma$ incident links), the degree of any $i\in\mc U_{1-h}$ is $\beta$ (one weight-$\beta$ link)}, and the total degree of nodes in $\mc U_h$ is $m(\gamma+\beta)+2l$, so that the aggregate centrality of $\mc S_h$ in \tcr{$\tilde{\mc G}_h$} is given by 
\be\label{pih}\tilde\pi_{\mc S_h}=(\gamma+\beta)m/(2(\gamma+\beta)m+2l)\,.\ee 

\tcb{Now, we recall a result in \cite[Ch.~6]{Durrettbook} stating that the connected Erdos-Renyi graph has conductance bounded away from $0$ with high probability as the network size grows large. This applies directly to the conductance $\Phi_h$ of $\mc G_h$, for $h=0,1$, while it can be shown to carry over to the conductance $\tilde\Phi_h$ of $\tilde{\mc G}_h$ upon verifying that $\gamma\ll\beta\ll1$ implies that the centralities and the transition probability between every $i,j\in\mc U_h$ are of the same order in $\mc G_h$ and $\tilde{\mc G}_h$, and that cuts in $\tilde{\mc G}_h$ separating subsets $\mc U\subseteq\{v_h\}\cup\mc U_{1-h}$ have bottleneck ratios bounded away from $0$. Hence, $\tilde\Phi_h\sim\Phi_h$ with high probability as $n$ grows large, and using the fact that the minimum degree in $\mc G_h$ is of order $\log m$, and the bound \eqref{conductance-bound}, one can show that the mixing time of $\tilde{\mc G}_h$ is of order $\log m$. 
By combining this with \eqref{pih}, it follows that, with high probability as $m$ grows large, the network $\hat{\mc G}_h$, $h=0,1$, is highly fluid if $\gamma\ll\beta\ll1$. }

\begin{figure}
\begin{center}
\includegraphics[height=3.5cm]{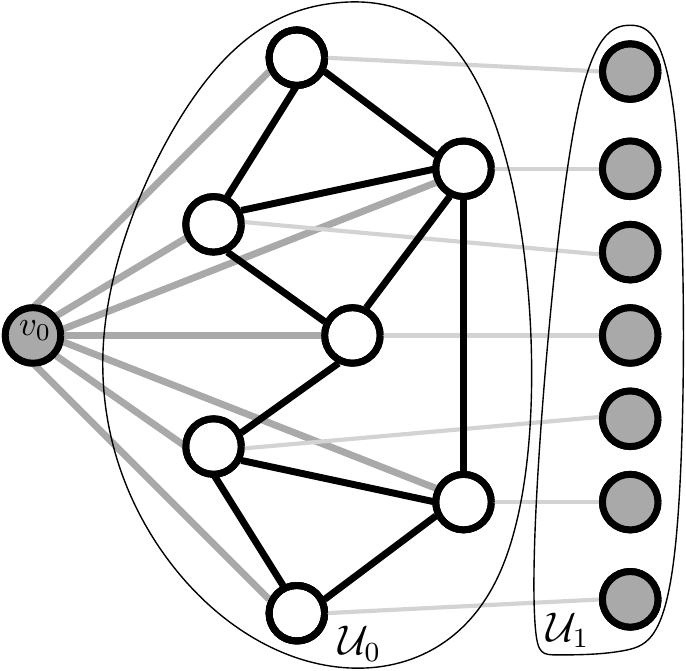}
\end{center}
\caption{\tcb{\label{fig:G0} The graph $\tilde{\mc G}_0$ obtained from $\mc G$ by removing node $v_1$ along with its incident links and all links connecting pairs of nodes in $\mc U_1$ and making links incoming in $v_0$ bidirectional.}}
\end{figure}

Finally, let $x$ and $\hat x^{(h)}$, for $h=0,1$, be the solutions of the Laplace equations respectively on $\mc G$ with boundary conditions $x_{v_0}=0$, $x_{v_1}=1$, and on $\hat{\mc G}_h$ with boundary conditions $\hat x^{(h)}_{v_h}=h$ and $\hat x^{(h)}_i=x_i$ for all $i\in\mc U_{1-h}$. Observe that $\hat x^{(h)}$ coincides with the restriction of $x$ on $\mc V\setminus\{v_{1-h}\}$. Then, Theorem \ref{theo:homog} implies that influence in $\hat{\mc G}_h$ is homogeneous, so that all but a vanishing fraction of nodes $i$ in  $\mc U_0$ and $\mc U_1$ have equilibrium state $\hat x_i=x_i$ close to $y_0$ and $y_1$, respectively. On the other hand, Proposition \ref{proposition:polarization} implies that, if $\gamma\ll\beta$, then $y_1-y_0\to0$, so that influence is homogeneous on the whole network $\mc G$.

\section{Conclusion}Simple and deep at the same time: two features that Jan Willems considered central in science can be well appreciated in the theory of distributed averaging. This paper has presented  some fundamental results for distributed averaging systems in a unified framework, giving a novel coherent perspective to classical material together with new generalizations.  The role of the electrical network interpretation in providing insight into the equilibrium analysis has been highlighted and 
some advanced material on the transition between homogeneous influence and polarization
has been presented.
 
Challenging problems for future research include: more complex heterogeneous networks,  robustness to perturbations \cite{Como.Fagnani:2015}, interconnections of higher order or nonlinear systems (e.g., coupled oscillators).

\section*{References}

\bibliographystyle{elsarticle-num} 
\bibliography{influence}

\end{document}